\documentclass{article}
\usepackage{times}
\usepackage{soul}
\usepackage{url}
\usepackage[hidelinks]{hyperref}
\usepackage[utf8]{inputenc}
\usepackage[small]{caption}
\usepackage{graphicx}
\usepackage{amsmath}
\usepackage{amsthm}
\usepackage{booktabs}
\usepackage{algorithm}
\usepackage{algorithmic}
\usepackage{natbib}
\usepackage{arxiv}
\renewcommand{\headeright}{}
\renewcommand{\undertitle}{}

\usepackage{amsmath,amsthm,amsfonts,amssymb,bm,bbm}
\usepackage{xspace}
\usepackage{nicefrac}
\usepackage{cleveref}
\usepackage{breqn}
\usepackage{xcolor}

\usepackage{tikz}
\usepackage{subcaption}

\newtheorem{theorem}{Theorem}
\newtheorem{lemma}{Lemma}

\newtheorem{proposition}{Proposition}

\theoremstyle{definition}
\newtheorem{definition}{Definition}

\renewcommand{\ge}{\geqslant}
\renewcommand{\geq}{\geqslant}
\renewcommand{\le}{\leqslant}

\DeclareMathOperator*{\argmin}{argmin}
\newcommand{\set}[1]{\left\{#1\right\}}

\renewcommand{\vec}{\bm}
\newcommand{\id}{\mathbbm{1}}
\newcommand{\bbN}{\mathbb{N}}
\newcommand{\bbR}{\mathbb{R}}

\newcommand{\e}{\varepsilon}
\newcommand{\V}{\mathcal{V}}

\renewcommand{\vec}[1]{#1}

\newcommand{\vsigma}{\sigma}
\newcommand{\vals}{{\vec{v}}}

\DeclareMathOperator{\dist}{dist}
\DeclareMathOperator{\SW}{sw}
\DeclareMathOperator{\MMS}{MMS}
\newcommand{\E}{\mathbb{E}}
\newcommand{\bigfloor}[1]{\left\lfloor#1\right\rfloor}
\newcommand{\bigceil}[1]{\left\lceil#1\right\rceil}
\newcommand{\I}{{\mathcal{I}}}

\title{Fair and Efficient Resource Allocation\\with Partial Information\thanks{A preliminary version appeared in the proceedings of the 30th International Joint Conference on Artificial Intelligence (IJCAI-21).}}

\renewcommand{\shorttitle}{Fair and Efficient Resource Allocation with Partial Information}

\author{
Daniel Halpern\\
Harvard University\\
\texttt{dhalpern@g.harvard.edu}
\And
Nisarg Shah\\
University of Toronto\\
\texttt{nisarg@cs.toronto.edu}
}

\date{}

\begin{document}
\sloppy
\maketitle

\begin{abstract}
We study the fundamental problem of allocating indivisible goods to agents with additive preferences. We consider eliciting from each agent only a ranking of her $k$ most preferred goods instead of her full cardinal valuations. We characterize the value of $k$ needed to achieve envy-freeness up to one good and approximate maximin share guarantee, two widely studied fairness notions. We also analyze the multiplicative loss in social welfare incurred due to the lack of full information with and without the fairness requirements. 
\end{abstract}

\section{Introduction}\label{sec:intro}
The theory of fair division studies how goods (or bads) should be fairly divided between individuals (a.k.a.\ agents) with different preferences over them. While the pioneering fair division research in economics, starting with the work of \citet{Stein48}, focused on \emph{divisible} goods which can be split between the agents, a significant body of recent research within computer science has focused on allocation of \emph{indivisible} goods~\cite{BCM16}.

Suppose we wish to partition a set of indivisible goods $M$ among a set of agents $N$. In doing so, we would like to take the agents' preferences into account; thus, the first step is to decide how to represent these subjective opinions over the possible bundle of goods the agent could receive. Some of the early work on fair division uses complete ordinal rankings: agents can have a nearly-arbitrary ordering over all $2^{|M|}$ subsets of the goods. Although of theoretical interest, as the number of goods grows, this domain quickly becomes too expressive and often leads to methods that are computationally infeasible (see, e.g.~\cite{HP02}). On the other end, another common approach is to allow agents to express ordinal preferences over the $|M|$ singleton subsets and extend these to ordinal preferences over all possible bundles (see, e.g.,~\cite{BK05,BEF03,aziz2015fair}). However, this suffers from the opposite problem and can often be too restrictive.  

Recent work has thus focused on a different option, \emph{additive cardinal preferences}. This preference domain is popular as it offers a sweet spot between simplicity and expressiveness. Here, each agent $i$ places a non-negative value $v_i(g)$ on each good $g$ and her value for a bundle of goods $S \subseteq M$ is assumed to be the sum of her values for the individual goods in $S$, i.e., $\sum_{g \in S} v_i(g)$. Theoretically, this valuation class gives way to algorithms achieving strong fairness guarantees~\cite{ABM16,CKMP+19,CGM20,GHSS+18,GT20}. Practically, additive valuations are much simpler to elicit than fully combinatorial valuations, which has led to their adoption by popular fair division tools such as Spliddit and Adjusted Winner.\footnote{\href{http://www.spliddit.org}{www.spliddit.org}, \href{https://pages.nyu.edu/adjustedwinner/}{www.nyu.edu/projects/adjustedwinner/}}

However, expressing additive valuations still requires placing an exact numerical value on each good, which can sometimes be difficult or infeasible. An interesting tradeoff can be achieved by eliciting ordinal preferences from the agents, but viewing them as partial information regarding underlying cardinal preferences. This idea originates from the related field of voting theory, where a growing body of work on the \emph{distortion} framework uses ordinal preferences of voters over candidates as means to pick a candidate approximately maximizing social welfare according to the underlying cardinal preferences~\cite{PR06,BCHL+15,CNPS17,MPSW19,MPW20,Kem20,ABFV20}.

In this paper, we focus on eliciting from each agent a ranking of her $k$ most preferred goods (i.e., a prefix of her preference ranking over the goods). A system designer deliberating on whether to use such partial information over traditional cardinal valuations may immediately be interested in the \emph{price} of the missing information. In line with the aforementioned work, we analyze distortion in the context of fair division, i.e., the worst-case (multiplicative) loss in social welfare --- the sum of the values that agents place on their own bundles --- incurred due to the missing information. 

In addition, we are also interested in achieving qualitative fairness guarantees; it is, after all, \emph{fair} division. Two popular guarantees for allocation of indivisible goods are envy-freeness up to one good (EF1)~\cite{LMMS04,Bud11} and approximate maximin share guarantee (MMS)~\cite{KPW18}, which we define in \Cref{sec:model}. With access to agents' full preference rankings over the goods, it is known that EF1 can be achieved via the round robin algorithm~\cite{LMMS04,CKMP+16}, under which agents take turns picking goods in a cyclic fashion. For MMS, \citet{ABM16} show that, using just ordinal preferences over the goods, it is impossible to guarantee better than a $\nicefrac{1}{H_n}$ approximation of MMS, where $H_n = \Theta(\log n)$ is the $n^{\text{th}}$ harmonic number and $n$ is the number of agents; in contrast, given additive cardinal preferences, even $\nicefrac{3}{4}$-MMS can be achieved~\cite{GHSS+18,GT20}. What is the best MMS approximation that can be achieved given agents' preference rankings over all the goods? More generally, if we are only given agents' preference rankings over their $k$ most preferred goods, for what values of $k$ can we achieve EF1 and approximate MMS? What distortion do we incur if, in addition to the missing cardinal information, we also impose these fairness requirements? We answer these questions in our work.

\subsection{Our Contribution}\label{sec:contrib}
A bit more formally, a deterministic (resp. randomized) ordinal allocation rule takes as input the partial preference rankings (where each agent provides a ranking of her $k$ most preferred goods) and returns an allocation (resp. a distribution over allocations) of the goods to the agents. The distortion of the rule is the ratio of the maximum social welfare of any allocation to the (expected) social welfare of the allocation returned, in the worst case over all problem instances in a family. As is common in the literature on distortion, we assume normalized valuations: the total value each agent places on all goods is normalized to $1$. We are interested in two questions. First, how much information is needed to achieve certain fairness guarantees? Second, what is the best distortion of any deterministic or randomized ordinal allocation rule with or without a fairness constraint?\footnote{For a randomized rule, we require that the fairness constraint be met by \emph{all} allocations in the support of the distribution returned.} 

Our results answer these questions for all values of $k$, but for simplicity, we summarize the results for when complete rankings are given ($k=m$) in \Cref{fig:summary}. Without any fairness constraint, the simple deterministic rule that simply allocates all the goods to a single agent achieves distortion $n$. We show that not even a randomized rule with access to complete rankings can achieve distortion better than $n$. 

Next, we consider two fairness requirements: envy-freeness up to one good (EF1) and approximate maximin share (MMS). EF1 is known to be achievable given complete rankings ($k=m$). We characterize the exact value of $k$ needed to achieve EF1. For MMS, we derive almost tight bounds the best possible approximation as a function of $k$. For the case of complete rankings ($k=m$), our results show that $\nicefrac{1}{(2H_n)}$-MMS is achievable, almost matching the asymptotic upper bound of $\nicefrac{1}{H_n}$ due to \citet{ABM16}. Thus, we establish, for the first time, that the best approximation to MMS given ordinal preference information scales logarithmically in the number of agents. 

We also show that when ordinal allocation rules are required to guarantee EF1 or $\alpha$-MMS for $\alpha > 0$, deterministic rules face $\Omega(n^2)$ distortion while randomized rules face $\Omega(n)$ distortion, and matching upper bounds can be derived (in case of MMS, along with best-known $\alpha$-MMS approximation). 
Our distortion upper bounds are achieved through efficient algorithms. In \Cref{sec:impossible-fairness}, we also show that various other fairness guarantees studied in the literature cannot be achieved given just ordinal preference information, even with complete rankings.

\begin{figure}
  \begin{minipage}[b]{0.60\linewidth}
  	\centering
    \begin{tabular}[b]{|c|c|c|}
		\hline 
		Fairness	 & Det			 & Rand \\ \hline
		None		 & $n$			 & $n$  \\
		EF1			 & $\Theta(n^2)$ & $\Theta(n)$  \\
		$\alpha$-MMS & $\Theta(n^2)$ & $\Theta(n)$  \\
		\hline
	\end{tabular}
  \end{minipage}%
  \begin{minipage}[b]{0.30\linewidth}
    \centering
\begin{tikzpicture}
		\tiny
		\tikzstyle{main_node} = [circle, draw, minimum size=4em,inner sep=.3em, align=center]

		\node[main_node] (1) at (0, 0) {Cardinal};
		\node[main_node] (2) at (1.2, -.9) {Cardinal\\ + $P$};
		\node[main_node] (3) at (0, -1.8) {Ordinal\\ + $P$};
		
		\draw[->] (1)--(2) node[midway, above right] {$\Theta(\sqrt{n})$};
		\draw[->] (1)--(3) node[midway, left] {$\Theta(n^2)$};
		\draw[->] (2)--(3) node[midway, below right] {$\Theta(n^2)$};
		
	\end{tikzpicture}
\end{minipage}
\caption{The table on the left summarizes the optimal distortion for deterministic and randomized rules with access to the complete rankings $(k = m)$. Note that $\alpha$-MMS is achievable for $\alpha=\nicefrac{1}{2H_n}$, but not for $\alpha > \nicefrac{1}{H_n}$. However, the distortion lower bounds hold for any $\alpha > 0$. The diagram on the right shows the worst-case ratio of social welfare between pairs of settings from the following three: cardinal valuations given, cardinal valuations given but property $P$ required, ordinal preferences given but property $P$ required. The diagram holds for both $P \in \set{\text{EF1}, \alpha\text{-MMS}}$ and the top-right arrow, the price of fairness $P$, is due to {\protect\citet{BBS20}}.} 
\label{fig:summary}
\end{figure}
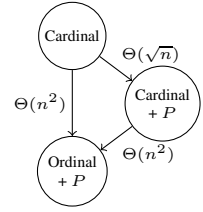

\subsection{Related Work}\label{sec:related}

There has been a substantial amount of work on using ordinal preferences in fair allocation of indivisible goods. For example, \citet{aziz2015fair} consider the question of checking the existence of allocations that possibly or necessarily satisfy certain fairness guarantees such as envy-freeness given only ordinal preferences of the agents over the goods. \citet{bouveret2010fair} study similar questions, but given partial ordinal preferences of the agents over \emph{bundles of goods}. 

Some of the work does not assume any underlying cardinal preferences; instead, it aims to obtain guarantees defined directly in terms of the ordinal preferences. For example, \citet{BaumeisterBLNNR17} and \citet{nguyen2017approximate} use the so-called scoring vectors to convert agents' ordinal preferences into numerical proxies for their utility and then consider maximizing various notions of social welfare or guaranteeing various fairness properties in terms of such utilities. 

Another related line of work uses ordinal allocation rules (such as picking sequence rules) in settings with cardinal valuations. For example, \citet{aziz2016welfare} focus on the complexity of checking what social welfare such rules can possibly or necessarily achieve. \citet{ABM16} seek to use picking sequence rules to obtain approximation of the maximin fair share guarantee; indeed, as mentioned earlier, we settle a question left open in their work. However, their main focus is on ensuring truthfulness, i.e., preventing agents from manipulating their preferences. Manipulations under picking sequence rules have received significant attention~\cite{aziz2017equilibria,aziz2017complexity}. 

\section{Model}\label{sec:model}

For $j \in \bbN$, let $[j] = \set{1, \ldots, j}$. Let $N = [n]$ be a set of \emph{agents} and $M = [m]$ be a set of \emph{goods}. Each agent $i$ is endowed with a \emph{valuation} function  $v_i: 2^M \to \bbR_{\ge 0}$, which is \emph{additive}: $v_i(S) = \sum_{g \in S} v_i(\set{g})$ for all $i \in N, S \subseteq M$; and \emph{unit-sum}: $v_i(M) = 1$ for all $i \in N$. To simplify notation, we write $v_i(g)$ instead of $v_i(\set{g})$ for a good $g \in M$. We refer to $\vals = (v_1, \ldots, v_n)$ as the \emph{valuation profile}. 

For $k \in [m]$, a \emph{top-$k$ ranking} $\sigma_i$ of agent $i$ is a ranking of agent $i$'s $k$ most valuable goods (ties broken arbitrarily). We say that a good is \emph{ranked} by an agent if it appears in their top-$k$ ranking and \emph{unranked} otherwise. We refer to $\vsigma = (\sigma_1, \ldots, \sigma_n)$ as the \emph{top-$k$ preference profile} (or, simply, preference profile). Note that the value of $k$ is the same for all agents. When $k = m$, we refer to these as complete rankings. We say that $v_i$ is consistent with $\sigma_i$, denoted $v_i \triangleright \sigma_i$, if $v_i(g) \ge v_i(g')$ for all $g, g' \in M$ such that either $g \succ_{\sigma_i} g'$ if both $g$ and $g'$ are ranked or $g$ is ranked and $g'$ is unranked. We say that $\vals$ is consistent with $\vsigma$, denoted $\vals \triangleright \vsigma$, if $v_i \triangleright \sigma_i$ for each $i \in N$.

We are interested in taking as input $(N, M, k, \vsigma)$, which we refer to as an \emph{instance}, and finding an allocation of the goods to the agents. For a set of goods $S \subseteq M$ and $\ell \in \bbN$, let $\Pi_\ell(S)$ denote the set of ordered partitions of $S$ into $\ell$ bundles. An allocation $\vec{A} = (A_1, \ldots, A_n) \in \Pi_n(M)$ is a partition of the goods into $n$ bundles, where $A_i$
is the bundle allocated to agent $i$. Under this allocation, the \emph{utility} to agent $i$ is $v_i(A_i)$.
Given a valuation profile $\vals$, the \emph{social welfare} of an allocation $\vec{A}$ is $\SW(\vec{A}, \vals) = \sum_{i \in N} v_i(A_i)$; we simply write $\SW(\vec{A})$ when the valuation profile
$\vals$ is clear from the context.

We will use $\I$ to denote a family of instances. We will use $\I^R$ to denote the family of instances in which relation $R$ over $k$, $n$, and $m$ is satisfied. For example, $\I^{k=m}$ is the family of instances with complete rankings and $\I^{k \ge n - 1}$ is the family of instances with rankings of at least $n-1$ goods.

A (randomized) \emph{ordinal allocation rule} (hereinafter, simply a rule) $f$ for a family of instances $\I$ takes an instance $(N,M,k,\vsigma)$ from $\I$ --- for simplicity, we refer to $\vsigma$ as the sole input to $f$--- and returns a distribution over the set of allocations $\Pi_n(M)$.
We say that $f$ is \emph{deterministic} if it always returns a distribution with singleton support. We will sometimes refer to a distribution over allocations as a \emph{randomized allocation}. The \emph{distortion} of an ordinal allocation rule $f$ with respect to a family of instances $\I$, denoted $\dist^\I(f)$, is the worst-case approximation ratio it
provides to the social welfare over all instances of $\I$:
\[
    \dist^{\I}(f) = \sup_{(N, M, k, \vsigma) \in \I} \sup_{\vals: \vals \triangleright \vsigma} \frac{\max_{\vec{A} \in \Pi_n(M)} \SW(\vec{A}, \vals)}{\E[\SW(f(\vsigma),\vals)]},
\]
where the expectation is over possible randomization in $f$. When $\I$ is clear from the context, we may drop it from the notation. Note that if $\I_1 \subseteq \I_2$, then $\dist^{\I_1}(f) \le \dist^{\I_2}(f)$. Following prior work and to help compare our distortion bounds to the known price of fairness bounds (see \Cref{fig:summary}), we provide distortion bounds parametrized by the number of agents $n$. 
We are interested in the lowest distortion that deterministic and randomized ordinal allocation rules can achieve.

A \emph{fairness property} $P$ maps every instance $I = (N,M,k,\vsigma)$ to a (possibly empty) set of allocations $P(I)$; every allocation in $P(I)$ is said to satisfy $P$ in instance $I$.
Often, fairness properties are defined in terms of agent valuations rather than rankings. In this case, an allocation is said to satisfy $P$ in instance $I$ only if it is satisfied by $P$ for all valuations consistent with $\vsigma$.
We say that a rule $f$ satisfies property $P$ if for all $\vsigma$,
every allocation in the support of $f(\vsigma)$ satisfies $P$.
We are also interested in determining whether ordinal allocation rules can satisfy prominent fairness properties, and when they can, determining the lowest possible distortion they can achieve subject to such properties.

Given an instance $I = (N,M,k,\vsigma)$ with valuations $\vals$, we are interested in the following fairness properties.

\begin{definition}[EF1]\label{def:ef1}
    An allocation $\vec{A}$ is called \emph{envy-free up to one good (EF1)} if for every pair of agents
    $i$, $j$, either $v_i(A_i) \ge v_i(A_j)$ or there exists a good $g \in A_j$ such that $v_i(A_i) \ge v_i(A_j \setminus \set{g})$.
\end{definition}

\begin{definition}[Balancedness]\label{def:balanced}
    An allocation $\vec{A}$ is called \emph{balanced} if $|A_i| - |A_j| \le 1$ for all $i, j \in N$, i.e., if the agents receive approximately an equal number of goods.
\end{definition}

\begin{definition}[MMS]\label{def:mms}
    The \emph{maximin share} of agent $i$ is 
    \[
        \MMS_i = \max_{\vec{A} \in \Pi_n(M)}\min_{A_j \in \vec{A}} v_i(A_j).
    \]
    Given $\alpha \in [0,1]$, an allocation $\vec{A}$ is called \emph{$\alpha$-maximin share fair} ($\alpha$-MMS) if $v_i(A_i) \ge \alpha \cdot \MMS_i$ for all agents $i \in N$. When $\alpha=1$, we simply say that $\vec{A}$ is an MMS allocation.
\end{definition}

\section{Distortion of Ordinal Allocation Rules}\label{sec:no-fair}

We begin by analyzing the lowest distortion that deterministic and randomized ordinal allocation rules can achieve in the absence of any fairness requirement. This precisely captures value of cardinal preference information, or the loss incurred in social welfare due to having only ordinal preferences. 

Even without any preference information ($k=0$), a trivial \emph{deterministic} rule that allocates all the goods to an arbitrary single agent achieves distortion $n$: indeed, the social welfare of such an allocation is $1$, while the maximum social welfare cannot be larger than $n$ since valuations are unit-sum ($v_i(M) = 1$ for all $i \in M$). We show that not even \emph{randomized} ordinal allocation rules with access to complete rankings ($k=m$) can achieve lower distortion.
\begin{theorem}
    \label{thm:all-rules-lower}
    There exists a deterministic ordinal allocation rule with distortion $n$ for the family $\I^{k \ge 0}$. On the other hand, no randomized ordinal allocation rule achieves distortion lower than $n$ even for the restricted family of $\I^{k = m}$. 
\end{theorem}
\begin{proof}
	Given the observation prior to the theorem statement, we only need to show that the distortion of an arbitrary randomized ordinal allocation rule $f$, even with access to the complete rankings, is at least $n$. 
	
	Fix an integer $x$, and let $\e = \nicefrac{1}{x}$. Let us consider a preference profile $\vsigma$ over $x^n$ goods, under which every agent has the same preference ranking: lower indexed goods are preferred to higher indexed goods. We construct a family of consistent valuations $\V$, and show that regardless of the randomized allocation returned by $f$ given $\vsigma$, the worst-case loss in social welfare across just this family of valuations already approaches $n$ as $x$ grows large. 
	
	We note that our construction below is similar to the one used by \citet{BLMS19} for lower bounding the social welfare loss incurred by a specific deterministic ordinal allocation rule, but our analysis significantly more intricate as it applies to all randomized ordinal allocation rules. 
	
	For $\ell \in [n]$, we say that an agent is of type $t_\ell$ if they like goods $1$ through $x^\ell$ equally, at value $\nicefrac{1}{x^\ell}$ each, and have value $0$ for the remaining goods. Let $T = \set{t_1, \ldots, t_n}$ be the set
    of all types. The family $\V$ consists of valuation profiles under which there is exactly one agent of each type. Such valuations can be represented by a bijection $\tau: [n] \mapsto T$ mapping agents to types.

    Consider the partition of goods $W = (W_1, \ldots, W_n)$ such that $W_1 = \set{1,\ldots,x}$, and $W_\ell = \set{x^{\ell-1}+1,\ldots,x^\ell}$ for $2 \le \ell \le n$. First, note that under every valuation in $\V$, there exists an allocation with social welfare at least $(1-\e)n$. This is achieved by assigning each $W_\ell$ to the agent of type $t_\ell$, which gives them value at least $(\nicefrac{1}{x^\ell}) \cdot |W_\ell| \ge (\nicefrac{1}{x^\ell}) \cdot (x^\ell - x^{\ell-1}) \ge 1-\nicefrac{1}{x}=1-\e$, resulting in social welfare at least $(1- \e)n$. In contrast, we show that under some valuation in $\V$, the randomized allocation returned by $f$ generates poor social welfare.

    First, note that under any valuation in $\V$, all agents are indifferent between the goods in $W_\ell$, for each $\ell \in [n]$. Hence, we can succinctly describe a randomized allocation by the
    expected fraction of goods from $W_\ell$ assigned to the different agents, for each $\ell \in [n]$. In other words, we can describe a randomized allocation $\mathcal{A}$ by an $n \times n$ matrix $X$ such that $X_{i,\ell}$
    is the expected fraction of goods from $W_\ell$ assigned to agent $i$. 
    Note that for the randomized allocation to be feasible, it must be the case that $\sum_{i=1}^n X_{i,\ell} = 1$ for all $\ell \in [n]$.

    Let $X$ be the matrix corresponding to the randomized allocation $\mathcal{A}$ returned by $f$ on input $\vsigma$ constructed above. Our goal is to find a valuation $\vals \in \V$, characterized by a bijection $\tau: [n] \to T$, under which the expected social welfare of $\mathcal{A}$ is not much more than $1$. First, we claim that if agent $i$ is of type $t_\ell$, then their expected utility under $\mathcal{A}$ is at most $X_{i,\ell} + \e$. To see this, note that agent $i$ is in expectation receiving an $X_{i,\ell}$ fraction of the goods from $W_\ell$, and $W_\ell$ consists of at least a $1-\e$ fraction of the goods they like, as argued above. Hence, the expected utility to agent $i$ under such an allocation is at most $X_{i,\ell}$ from the goods in $W_\ell$ and at most $\e$ from the goods outside of $W_\ell$, which is at most $X_{i,\ell}+\e$ in total. Thus, the expected social welfare of $A$ is at most $\sum_{i,\ell : \tau(i) = t_\ell} X_{i,\ell} + n \e$. Our goal is to show that there exists a bijection $\tau$ for which this quantity is at most $1+n\e$.
    
    Note that if we choose $\tau$ uniformly at random, then the expected value of $\sum_{i,\ell : \tau(i) = t_\ell} X_{i,\ell}$ is equal to $\nicefrac{1}{n} \cdot \sum_{i \in [n], \ell \in [n]} X_{i,\ell} = 1$. Hence, there must exist a bijection $\tau$ for which $\sum_{i,\ell : \tau(i) = t_\ell} X_{i,\ell} \le 1$. Under the corresponding valuation profile, the expected social welfare of $\mathcal{A}$ is at most $1+n\e$. This shows that the distortion of $f$ is at least $\frac{(1-\e)n}{1+n\e}$. As $x$ approaches $\infty$, $\e$ approaches $0$, which establishes that the distortion of $f$ is at least $n$, as desired.
\end{proof}

\subsection{Fairness Lower Bounds}\label{subsec:fair-dist-lower}

In this section, we analyze the lowest distortion that ordinal allocation rules can achieve when they are required to satisfy some fairness constraints. This captures the \emph{combined price} of the lack of cardinal preference information and the imposition of fairness constraints. \Cref{fig:summary} contrasts this with the sole price of the former analyzed in \Cref{sec:no-fair} and the sole price of the latter from known results in the literature. Perhaps not surprisingly, it turns out that the two together lead to a much greater loss in social welfare than each individually. 

Another consequence of our results is that while randomized ordinal  rules are no more powerful than deterministic ones in the absence of any fairness requirements (\Cref{thm:all-rules-lower}), imposing fairness requirements makes their powers diverge. 

Keeping aside the question of distortion, we are also interested in determining which fairness properties ordinal allocation rules can satisfy. A negative answer can be interpreted as a qualitative price of the lack of cardinal preferences. 

We begin by establishing a lower bound on the distortion of deterministic ordinal allocation rules that holds when any fairness property from a broad class is imposed, even with access to complete rankings; later, we argue that the fairness properties of our interest belong to this class. Recall that we require the allocation returned by the rule to satisfy the fairness property, regardless of the unobserved cardinal valuations (consistent with the observed ordinal preferences).
\begin{theorem}
    \label{thm:n2-lower}
    Let $P$ be a fairness property such that when the number of goods equals the number of agents, for every preference profile $\vsigma$, an allocation satisfies $P$ for all valuations consistent with $\vsigma$ if and only if each agent receives a single good. Then, the distortion of every deterministic ordinal allocation rule satisfying $P$ is $\Omega(n^2)$ for the family $\I^{k = m}$. 
\end{theorem}
\begin{proof}
    Fix such a fairness property $P$, a number of agents $n$, and a deterministic ordinal allocation rule $f$ on $\I^{k = m}$ (i.e., only taking complete rankings) satisfying $P$. First, let us suppose that $n$ is even. We construct an instance with $n$ goods, that is, with $m = n$. We split the goods into three different categories and construct a preference profile $\vsigma$ as follows. The first category consists of a single good $g^*$ that is ranked highest by all agents. The next category consists of $\nicefrac{n}{2}$ goods labeled $g_{\set{1,2}},g_{\set{3,4}},\ldots,g_{\set{n-1,n}}$. For each $\ell \in [\nicefrac{n}{2}]$, good $g_{\set{2\ell - 1, 2\ell}}$ is ranked second by both agents $2\ell-1$ and $2\ell$. The final category consists of the remaining $\nicefrac{n}{2}-1$ goods. The construction above identifies the two most preferred goods for all agents; their preference rankings from the third rank onward can be arbitrary. 

    Let $A$ be the allocation returned by $f$ given $\vsigma$. By the assumption of the theorem statement, each agent must receive exactly one good in $A$. Without loss of generality, let us assume that agent $1$ receives $g^*$. In addition, for each $\ell \in \set{2,\ldots,\nicefrac{n}{2}}$, at least one of agents $2\ell-1$ and $2\ell$ does not receive good $g_{\set{2\ell - 1, 2\ell}}$; without loss of generality, assume that agent $2\ell-1$ does not receive it. Let us construct a consistent valuation profile as follows:
    \begin{itemize}
        \item Agent $1$ has value $\nicefrac{1}{n}$ for each good.
        \item Agent $2$ has value $1$ for $g^*$ and $0$ for all other goods.
        \item For $\ell \in \set{2,\ldots,\nicefrac{n}{2}}$, agent $2\ell-1$ has value $\nicefrac{1}{2}$ for $g^*$, $\nicefrac{1}{2}$ for $g_{\set{2\ell-1,2\ell}}$, and $0$ for all other goods; and agent $2\ell$ has value $1$ for $g^*$ and $0$ for all other goods.
    \end{itemize}
    Under $\vec{A}$, the only agent receiving positive utility is agent $1$, who receives utility $\nicefrac{1}{n}$. Therefore,
    the social welfare is $\nicefrac{1}{n}$. In contrast, consider the allocation that gives $g^*$ to agent $2$, $g_{\set{2\ell-1,2\ell}}$ to agent $2\ell-1$ for each $\ell \in \set{2,\ldots,\nicefrac{n}{2}}$, and the remaining goods arbitrarily such that each agent receives a single good. It is easy to check that its social welfare is at least $1 + (\nicefrac{n}{2} - 1) \cdot \nicefrac{1}{2} = \nicefrac{n}{4} + \nicefrac{1}{2}$. Therefore, the distortion of $f$ is at least $(\nicefrac{n}{4} + \nicefrac{1}{2}) / (\nicefrac{1}{n}) \in \Omega(n^2)$.

    If $n$ is odd, we can construct the described instance with $n-1$ agents and $n-1$ goods, add a good ranked last by all agents, and add an agent whose preference ranking matches that of one of the other agents. Using similar arguments as above, regardless of the allocation $A$ chosen by $f$, we can construct a consistent valuation profile in which the social welfare of $A$ is     $\nicefrac{1}{n}$, while the optimal social welfare is at least $\nicefrac{(n - 1)}{4} + \nicefrac{1}{2}$, resulting in $\Omega(n^2)$ distortion. 
\end{proof}

Notice that in the proof of \Cref{thm:n2-lower}, we contrast the social welfare achieved by the rule satisfying $P$ against that of an allocation that assigns each agent a single good, thus also satisfying $P$. That is, the $\Omega(n^2)$ lower bound continues to hold even when comparing to the optimal social welfare \emph{subject to $P$}. On the other hand, our matching upper bounds presented later hold even when comparing to the optimal social welfare \emph{without any fairness constraints}.

\section{EF1}\label{sec:ef1} 
	We now turn our attention to EF1. We begin by fully characterizing the values of $k$, in relation to $n$ and $m$, for which we can achieve EF1. The rules we construct are based on \emph{picking sequence rules}. A picking sequence is simply a sequence of agents $p_1,\ldots,p_{\ell}$, where $\ell \le m$. It is a deterministic rule that works as follows: it first gives agent $p_1$ their favorite good, then gives agent $p_2$ their favorite good among the ones remaining, and so on, for $\ell$ steps. If $\ell < m$, we design a way --- different for each case --- for allocating the remaining goods. A well-known picking sequence rule is \emph{round robin}, which has the cyclic picking sequence $1,\ldots,n,1,\ldots,n,\ldots$ repeated for a total of $m$ steps.

	\begin{theorem}\label{thm:ef1-possible}
		With $n$ agents and $m$ goods, it is possible to guarantee EF1 using top-$k$ rankings if and only if
		\[
		k \ge \begin{cases}
		m - n, &\mbox{if } m \bmod n = 0;\\
		m - 2, &\mbox{if } m \bmod n = 1;\\
		m - (m \bmod n), &\mbox{if } m \bmod n > 1.
		\end{cases}
		\]
	\end{theorem}
	\begin{proof}
	Fix arbitrary $n$, $m$, and $k$. We begin with the lower bounds, showing that EF1 can only be achieved if $k$ is sufficiently large. All of our constructions have the same preference profile: all agents agree on which goods are in the top $k$, that is, they rank goods $g_1, \ldots, g_k$ in some order and do not rank the remaining $m - k$ goods. For a given allocation $A$, let $s_i = |A_i \cap \set{g_1, \ldots, g_k}|$ be the number of the top-$k$ goods received by agent $i$. Note that $k \ge \sum_{i \in N} s_i$. We use the following lemma.
	\begin{lemma}\label{lem:k-ef1}
		If agents agree on which goods are in the top $k$ and an allocation $A$ is EF1 for all consistent valuations, then $s_i \ge |A_j| - 1$ for all distinct agents $i, j \in N$.
	\end{lemma}
	\begin{proof}
	Consider a consistent valuation profile in which agent $i$ has zero value for the goods $A_i \setminus \set{g_1, \ldots, g_k}$ but equal value for all other goods (including all of the ones in $A_j$). If $s_i < |A_j| - 1$, EF1 would be violated for agent $i$. 
	\end{proof}
	
	Suppose there exists an allocation $A$ guaranteed to be EF1. We show that this implies $k$ is sufficiently large as per the theorem statement. Let $q \in \bbN$ and $r \in [n-1]$ be such that $m = qn + r$. Since $A$ is guaranteed to be EF1 given ordinal preferences, it must be balanced, that is, $|A_j| - |A_i| \le 1$ for all agents $i$ and $j$: using \Cref{lem:k-ef1}, we can see that $|A_i| \ge s_i \ge |A_j|-1$. In our case, this means that $r$ agents have bundles with size $q + 1$ and the remaining $n - r$ have bundles with size $q$. In the following, we use $k \ge \sum_{i \in N} s_i$. 
	\begin{itemize}
		\item Suppose $r = 0$, so $|A_i| = q$ for all agents $i$. By \Cref{lem:k-ef1}, $s_i \ge q - 1$ all agents $i$, so $k \ge (q - 1)n = m - n$.
		\item Suppose $r = 1$. Therefore, one bundle, without loss of generality $A_1$, has size $q + 1$, and all the others have size $q$. We have that $s_1 \ge q - 1$ and $s_i \ge q$ for all $i \ne 1$ by \Cref{lem:k-ef1}. This implies $k \ge qn - 1 = m - 2$.
		\item Suppose $r > 1$. As at least two agents have bundles of size $q + 1$, by \Cref{lem:k-ef1}, $s_i \ge q$ for all agents $i$. This implies $k \ge qn$.
	\end{itemize}	
	
	Next, we prove the upper bounds, showing that if $k$ is sufficiently large, EF1 can be guaranteed. To do this, we make modifications to the aforementioned round robin rule, which (with a picking sequence of length $m$) is known to guarantee EF1. Note that to run this rule, only ordinal information is needed but complete rankings (or at least $k \ge m - 1$) are needed. To work with round robin, we label the goods in the order they are chosen as follows $g_{1, 1}, g_{2, 1}, \ldots, g_{n, 1}, g_{1, 2}, \ldots, g_{\ell, t}$, where good $g_{i, j}$ is the $j^{\text{th}}$ good picked by agent $i$. We refer to $t$, the largest number of goods picked by any agent, as the number of ``rounds'' and goods $g_{i, t}$ for $i \in N$ as goods picked in the last round. Note that not all agents will necessarily pick a good in the last round. We make use of the following lemma.
	\begin{lemma}\label{lem:rr-permute}
		Suppose we run round robin to $m$ steps with access to complete rankings but reassign the goods received by agents in the last round such that no agent receives more than one good. Then the resulting allocation remains EF1.
	\end{lemma}
	\begin{proof}
	The proof is a slight modification of the classic proof that round robin is EF1. First, we will show the classic proof for completeness and then show how to modify it to achieve our desired result. Without loss of generality, we will add dummy goods $g_{\ell + 1, t}, \ldots, g_{n, t}$, which all agents value at zero, so that all agents receive the same number of goods. Note that the addition of dummy goods does not affect whether an allocation is EF1. We have that $A_i = \set{g_{i, j}|1 \le j \le t}$. Fix two agents $i$ and $i'$; we will show that $i$ does not envy $i'$ up to one good. First, note that $v_i(g_{i, j}) \ge v_i(g_{i', j + 1})$ for all $1 \le j \le t - 1$. This is because $g_{i', j + 1}$ was available on $i$'s $j^{\text{th}}$ pick, but $i$ preferred good $g_{i, j}$. Therefore, we have that
	\[
		v_i(A_i)
		\ge v_i(A_i \setminus \set{g_{i, t}})
		= \sum_{j = 1}^{t - 1} v_i(g_{i, j})
		\ge \sum_{j = 1}^{t - 1} v_i(g_{i', j + 1})
		= \sum_{j = 2}^{t} v_i(g_{i', j})
		= v_i(A_j \setminus \set{g_{j, 1}})
	\]

	as needed.
	
	Now for the modification. We have that $g_{i, t}$ may be replaced by any good received in the $t^{\text{th}}$ round and the above inequalities still hold as $g_{i, t}$ is simply ignored from $A_i$. Further, $g_{i', t}$ may be replaced by any good $g$ in the $t^{\text{th}}$ round and $v_i(g_{i, t-1}) \ge v_i(g)$ will still hold. This implies EF1 will still hold even after the reassignment.
\end{proof}

	All the rules we design will run round robin to at most $k$ steps and then assign the remaining goods in a way that can be accomplished without access to the remaining ordinal preferences. We will then argue that the remaining goods were assigned in a way such that they only permuted the last round of goods in a hypothetical allocation given by running round robin to $m$ steps with access to complete rankings. By \Cref{lem:rr-permute}, this implies EF1.
	
	First, suppose $m \bmod n = 0$ and $k \ge m - n$.
	The rule works as follows: it runs round robin for $m - n$ steps and assigns the remaining $n$ goods so that each agent receives exactly one. This is EF1 by \Cref{lem:rr-permute} as the resulting allocation could also have been computed by running round robin using complete rankings and reassigning the goods from the last round in the way that was arbitrarily chosen.
	
	Next, suppose $m \bmod n = 1$ and $k \ge m - 2$.
	The rule works in a very similar way. It runs round robin for $m - 2$ steps and assigns the remaining two goods to the last agent in the order (or if $m = 1$, just assigns the good to an arbitrary agent). This is EF1 by \Cref{lem:rr-permute} as the resulting allocation could have been computed by running round robin and giving the singular good of the final round to the $n^{\text{th}}$ agent in the order.
	
	Finally, suppose $m \bmod n > 2$ and $k \ge m - (m \bmod n)$. 
	As before, we run round robin for $m - (m \bmod n)$ steps and assign the remaining $m \bmod n$ goods such that each agent receives at most one. This is again EF1 by \Cref{lem:rr-permute}.
\end{proof}

Next, we will consider the best possible distortion of these rules. However, when analyzing the distortion of randomized rules, we will need the following observations.
Given a deterministic picking rule, one way to turn this into a randomized rule, known as its uniform variant, is as follows: pick a permutation of the agents $\tau$ uniformly at random, and then run the rule using this agent labeling.
Earlier, we observed that distortion $n$ can be achieved trivially by allocating all the goods to an arbitrary single agent. The following result shows that a more interesting class of rules --- \emph{uniform picking sequence rules} --- also achieves distortion $n$.
This result is a well-known folklore in the literature, but we present a proof for completeness. A proof for the special case of round robin method was given by \citet{FSV20}. 

\begin{proposition}
    \label{prop:uniform-picking-sequences}
    If a rule begins with a picking sequence of at least $cm$ steps for some constant $0 < c \le 1$, then the uniform variant has distortion in $[n, \nicefrac{n}{c}]$. In paticular, it has distortion $\Theta(n)$, and if $c = 1$, it has distortion $n$.
    
\end{proposition}
\begin{proof}
	Consider some rule $f$ that always runs a picking sequence of length $cm$.
    Fix an arbitrary $I = (N, M, \vsigma, k)$ with $m$ goods and underlying valuations $\vals$.
    Let the picking sequence be $p_{1}p_{2}\ldots p_{\ell}$ for $\ell \ge cm$. We show that the expected social welfare under the rule is at least $c$.
    Since the optimal social welfare is upper bounded by $n$, this gives the desired upper bound of $cn$ on distortion.

    To show the expected social welfare is at least $c$, we show that the expected utility to each agent $i$ is at least $\nicefrac{c}{n}$. Suppose agent $i$ has preference ranking $g_1 \succ \dots \succ g_m$. Because utilities are additive, we can decompose the expected utility to agent $i$ as the sum of the expected utilities received by agent $i$ from all $\ell$ picks. 
    
    Note that agent $i$ only receives a nonzero utility from the $\j^{\text{th}}$ pick if $\tau(p_j) = i$, where $\tau$ is the random permutation chosen by the rule. This happens with probability $1/n$, and with this probability, agent $i$ picks a good at least as valuable to her as $g_j$ (since one of $g_1,\ldots,g_j$ must be available during $j^{\text{th}}$ pick), thus receiving utility at least $v_i(g_j)$. Therefore, the expected utility of agent $i$ from the $j^{\text{th}}$ pick is at least $\nicefrac{1}{n} \cdot v_i(g_j)$ for each $j \le \ell$. Using linearity of expectation, we get that the expected utility to agent $i$ is at least $\sum_{j=1}^\ell \nicefrac{1}{n} \cdot v_i(g_j) \ge \nicefrac{1}{n} \cdot c \cdot v_i(M) = \nicefrac{c}{n}$, where the first transition holds because $g_1, \ldots, g_{\ell}$ account for at least a $c$ fraction of their total value.
\end{proof}

	Let $\I^{EF1}$ be the family of instances with $k$, $n$, and $m$ values satisfying the relation specified in \Cref{thm:ef1-possible} (i.e. for which it is possible to achieve EF1). Then, the best possible distortion subject to the requirement of achieving EF1 is as follows. 
	\begin{theorem}\label{thm:ef1-distortion}
		Among deterministic ordinal rules, all EF1 rules have unbounded distortion on $\I^{EF1} \cap \I^{k = 0}$ and the lowest possible distortion of an EF1 rule with respect to $\I^{EF1} \cap \I^{k \ge 1}$ is $\Theta(n^2)$. Among randomized ordinal rules, the lowest possible distortion of an EF1 rule on $\I^{EF1}$ is $\Theta(n)$. 
	\end{theorem}
	\begin{proof}
	First, let us consider $\I^{EF1} \cap \I^{k = 0}$. We will show it has unbounded distortion. This family includes an instance with two agents, $1$ and $2$, two items $g_1$ and $g_2$, and empty rankings. As shown in the proof of \Cref{thm:ef1-possible}, EF1 implies balanced, so any EF1 rule must give one item to each agent. Without loss of generality, suppose $g_1$ goes to agent $1$ and $g_2$ goes to agent $2$. It's consistent with the rankings that agent $1$ has value one for $g_2$ and zero for $g_1$, while agent $2$ has value one for $g_1$ and zero for $g_2$. This shows the distortion of the rule is unbounded.

	Next, let us consider $\I^{EF1} \cap \I^{k \ge 1}$. By \Cref{thm:n2-lower}, the distortion of any EF1 rule on this family must be $\Omega(n^2)$.
	We will show that $O(n^2)$ distortion is possible. To do this, we will find an EF1 rule that always gives agent $1$ value at least $\nicefrac{1}{3n}$.
	As the total welfare of any allocation is at most $n$, this will imply the rule has the desired distortion.
	We will split the rule into two cases, $m \le n$ and $m > n$.
	If $m \le n$, then any allocation such that each agent receives at most one good is EF1.
	As $k \ge 1$, we know which good is agent $1$'s favorite, so we will choose such an allocation where agent $1$ receives this good.
	The value that an agent has for the most valuable good is at least $\nicefrac{1}{m} \ge \nicefrac{1}{n} > \nicefrac{1}{3n}$ as needed. If $m > n$, we will run the associated EF1 algorithm described in the proof of \Cref{thm:ef1-possible}, and ensure agent $1$ is first in the order of the round robin portion of the rule. We claim that if round robin is run for $\ell$ steps, then agent $1$ receives at least $\frac{\ell}{mn}$ value. Since in all the EF1 rules, whenever $m > n$, round robin is run for at least $\nicefrac{m}{3}$ steps, this shows the desired result. To show the above claim, note that agent $1$'s top $\ell$ goods account for at least $\nicefrac{\ell}{m}$ value. Further, in each round of the round robin step, agent $1$ receives an item at least as valuable to them as the other $n - 1$ items, so receives at least $\nicefrac{1}{n}$ of their value of the top $\ell$ goods.
	
	For randomized rules, note that the $\Omega(n)$ distortion follows from \Cref{thm:all-rules-lower}. For $O(n)$, we will again split it into two cases, $m \le n$ and $m > n$. When $m \le n$, we will pick an assignment of the goods such that each agent receives at most one uniformly at random. This is EF1 as each agent receives at most one good. Further, each agent receives each good with probability $\nicefrac{1}{n}$, so the expected value of an arbitrary agent $i$ is $\sum_{g \in M} \nicefrac{1}{n} \cdot v_i(g) = \nicefrac{1}{n} \cdot \sum_{g \in M} v_i(g) = \nicefrac{1}{n}$. Hence, the expected social welfare is $1$, so distortion is at most $O(n)$.
	
	When $m > n$, we will run the associated EF1 rule from \Cref{thm:ef1-possible} in a uniform way, picking a permutation $\tau$ of the agents uniformly at random, and then running round robin with this agent order for the designated number of steps, finally distributing the remaining goods in the specified way. The key here is that, as $m > n$, we will always be running uniform round robin for at least $\nicefrac{m}{3}$ steps. As a straightforward consequence of \Cref{prop:uniform-picking-sequences}, this randomized rule will have distortion at most $3n \in O(n)$.
\end{proof}
\section{MMS}\label{sec:mms}

We now turn to our most technical results, which are regarding approximate maximin share (MMS) guarantee. Before we consider distortion subject to approximate MMS, we need to know what approximation to MMS is possible to achieve. Given full cardinal information, it is known that exact MMS cannot be achieved~\cite{KPW18}, but $\nicefrac{3}{4}$-MMS can~\cite{GHSS+18,GT20}. 

Given complete preference rankings ($k=m$), \citet{ABM16} show that it is not possible to achieve $\alpha$-MMS for $\alpha > \nicefrac{1}{H_n}$, where $H_n = \Theta(\log n)$ is the $n^{\text{th}}$ Harmonic number. On the opposite end, they only establish a weaker $\Omega(\nicefrac{1}{\sqrt{n}})$ lower bound, leaving open the question of what the best possible MMS approximation is given complete preference rankings. We settle this question by showing that the best possible MMS approximation for $k=m$ is $\Theta(\nicefrac{1}{H_n})$ (specifically, we derive a lower bound of $\nicefrac{1}{2H_n}$). We also extend the lower and upper bounds to the case of $k < m$. 

Our algorithm is similar to the one provided by \citet{ABM16} to achieve $\Omega(\nicefrac{1}{\sqrt{n}})$, but our improvement crucially relies on \Cref{lem:d-ineq}, which requires an intricate proof to achieve the desired bounds. Note that for $m \le n$, MMS can trivially be satisfied by giving each agent at most one good; thus, we focus on $m > n$. 
\begin{theorem}\label{thm:mms-possible}
	When $m > n$, the following hold. 
	\begin{itemize}
		\item If $k < n - 1$, we cannot achieve $\alpha$-MMS for any $\alpha > 0$.
		\item If $k = n - 1$, we can achieve $\frac{1}{\bigfloor{\frac{m - n + 2}{2}}}$-MMS, but no higher.
		\item If $k \ge n$, we can achieve $\alpha$-MMS for $\alpha = \frac{k - n + 1}{m - n + 1} \cdot \frac{1}{2H_n}$, but not for $\alpha > \frac{k}{H_n(m-n) - (m - k)}$.
	\end{itemize}
\end{theorem}
\begin{proof}
	We begin with the most interesting case, a lower bound for $k \ge n$. We borrow and build upon ideas from the proof of the $\Omega(\nicefrac{1}{\sqrt{n}})$ lower bound due to \citet{ABM16}.

    We construct a picking sequence rule achieving the desired MMS approximation. 
    Fix an agent. For now, suppose we are working with complete rankings, with $k = m$. Suppose the first time this agent appears in the picking sequence is at the $\ell^{\text{th}}$ position (we call this the agent's $0^{\text{th}}$ appearance) for some $\ell \le n$, and then, the agent's $j^{\text{th}}$ appearance occurs at or before position $(\ell + \bigfloor{j \cdot 2H_n (n - \ell + 1)})$ in the picking sequence for every $j$ (as long as this quantity does not exceed $m$). Then, we claim that the agent must be guaranteed at least $\nicefrac{1}{(2H_n)}$ fraction of her MMS value. To see this, note that the agent picks a good at least as valuable as her $\ell^{\text{th}}$ most favorite good in her $0^{\text{th}}$ appearance, and then an additional good at least as valuable as her $(\ell + \bigfloor{j \cdot 2H_n (n - \ell + 1)})^{\text{th}}$ favorite good in her $j^{\text{th}}$ appearance for each $j$. Let $S$ denote the total value the agent places on her $\ell-1$ most valuable goods. Then, this picking sequence guarantees the agent utility at least $\frac{1-S}{2H_n(n-\ell+1)}$. On the other hand, note that the MMS value of the agent is at most $\frac{1-S}{n-\ell+1}$; this is because regardless of how the agent partitions the goods into $n$ bundles, ignoring the (at most) $\ell-1$ bundles containing her $\ell-1$ most valuable goods, even the average value across the remaining (at least $n-\ell+1$) bundles is at most $\frac{1-S}{n-\ell+1}$. Hence, it follows that the agent is guaranteed at least a $\nicefrac{1}{2H_n}$ fraction of her MMS. 
    
    Now, instead of assuming we have complete rankings, suppose we just have top-$k$ for some $k \ge n$. In this case, suppose we run the above picking sequence for just $k$ steps. Then, rather than being guaranteed $\frac{1-S}{2H_n(n - \ell + 1)}$, the agent is only guaranteed $\frac{k - \ell + 1}{m - \ell + 1} \cdot \frac{1-S}{2H_n(n-\ell+1)}$. As their MMS value remains the same and $\ell \le n$, this agent is guaranteed $\frac{k - n + 1}{m - n + 1} \cdot \frac{1}{2H_n}$-MMS as needed. The remaining $m - k$ goods can be allocated arbitrarily.
    
    Our picking sequence gives a guarantee of this style to each agent, albeit for different values of $\ell$. In particular, for each agent $i \in [n]$, the picking sequence provides this guarantee with $\ell=i$. The construction is very simple. 

    \begin{enumerate}
        \item For $1 \le i \le n$ and $0 \le j \le \bigfloor{\frac{m - i}{2H_n(n - i + 1)}}$, we create the pair $(i, i + \bigfloor{j \cdot 2H_n(n - i + 1)})$, indicating that agent $i$'s $j^{\text{th}}$ appearance must occur at or before the position indicated in the second component --- we refer to this as the deadline.
        \item We sort the pairs with respect to their second coordinate.
        \item The first coordinates with respect to the above sorting are a prefix of the picking sequence.
        \item If the length of the above sequence is $m$, we are done; otherwise we arbitrarily assign the remaining picks.
        \item If $k < m$, we truncate the sequence to length $k$.
    \end{enumerate}

	The idea of Steps 2--4 is to produce a picking sequence that meets all the deadlines by using earliest-deadline-first scheduling. It is known that if all the deadlines can be met, then this greedy scheduling procedure is guaranteed to return a sequence meeting them. To show that all deadlines are met, we want to show that there are at most $d$ pairs introduced in Step 1 with the second coordinate (deadline) at most $d$, for all $d \le m$. Note that in particular, this implies that there are at most $m$ pairs in total, so Step 3 would not produce a sequence of length more than $m$.

    To prove this, let us first consider $d \le n$.
    Observe that the $1^{\text{st}}$ appearance deadline of any agent is at or after position $n+1$: this is because $i + \bigfloor{2H_n(n - i + 1)} \ge 1+\bigfloor{n-i+1} = n+1$ for all $i \in [n]$. This implies that the only pairs with deadline at most $n$ are the $n$ pairs of the form $(i,i)$ for $i \in [n]$ corresponding to the $0^{\text{th}}$ appearances of all the agents, which immediately implies the desired goal holds for all $d \le n$.
    
    Next, consider $d \ge n + 1$.The number of pairs for agent $i$ with the second coordinate at most $d$ is
    at most $1 + \bigfloor{\frac{d - i}{2H_n (n - i + 1)}}$. Therefore, the number of total pairs with second coordinate at most $d$ is at most
    $\sum_{i = 1}^n 1 + \bigfloor{\frac{d - i}{2H_n (n - i + 1)}} = n + \sum_{i = 1}^n \bigfloor{\frac{d - i}{2H_n (n - i + 1)}}$. Our goal is to show that this value is at most $d$, which is equivalent to the following lemma. 
    \begin{lemma}
        \label{lem:d-ineq}
        For all $n \in \bbN$ and for all $d \ge n + 1$, $\sum_{i = 1}^n \bigfloor{\frac{d - i}{2H_n (n - i + 1)}} \le d - n$.
    \end{lemma}
    \begin{proof}
    There are two cases, either $d \ge 2n$ or $d < 2n$.
    First suppose that $d \ge 2n$.
    Then,
    \begin{align*}
        \sum_{i=1}^n \bigfloor{\frac{d - i}{2H_n (n - i + 1)}}
        &\le \sum_{i=1}^n \frac{d - i}{2H_n (n - i + 1)}\\
        &\le \frac{d}{2H_n}\sum_{i=1}^n \frac{1}{n - i + 1}\\
        &= \frac{d}{2} \le d - n.
    \end{align*}

    Next, assume that $n + 1 \le d < 2n$. Let $r = d - n$.
    Note that $1 \le r < n$.
    The inequality in the lemma is therefore equivalent to (note the cancelling $+1$ and $-1$ in the numerator) $\sum_{i=1}^n \bigfloor{\frac{r + (n - i + 1) - 1}{2H_n (n - i + 1)}} \le r$. Note that summing over $i$ ranging from $1$ through $n$ is equivalent to summing over $n-i+1$ ranging from $1$ through $n$; substituting $i$ to denote $n-i+1$, we need to prove $\sum_{i=1}^n \bigfloor{\frac{r + i - 1}{2H_n \cdot i}} \le r$.

    Call the summand value $f(i) = \bigfloor{\frac{r + i - 1}{2H_n i}}$.
    First, we show that $f(i) \le r$ for all $i \in [n]$. For this, it is sufficient to show $r + i - 1 \le 2H_n \cdot i \cdot r$.
    We have that $r + i - 1 \le 2\max(i, r) \le 2H_n \cdot i \cdot r$ because $H_n$, $i$, and $r$ are all at least $1$.
    
	Let $g(j)$ be the number of values of $i$ such that $f(i) \ge j$; more formally, $g(j) = \sum_{i = 1}^n \id[f(i) \ge j]$.
    Now since $f(i) \le r \le n$, we have that
    $\sum_{i = 1}^n f(i) = \sum_{j = 1}^n g(j)$ as each agent $i$ is counted precisely $f(i)$ times in the RHS sum. Next, we bound $g(j)$ for each $j \in [n]$. To do this, we show that if $i$ is such that $f(i) \ge j$, then $i$ must be
    bounded by some value, $B$. This would imply that at most $B$ values of $i$ have $f(i) \ge j$, implying $g(j) \le B$. 
    \begin{align*}
        & f(i) = \bigfloor{\frac{r + i - 1}{2H_n \cdot i} }\geq j\\
        \implies & r + i - 1 \geq 2 H_n \cdot i \cdot j \tag{$2H_n \cdot i > 0$}\\
        \implies & r - 1 \geq (2H_n j - 1) \cdot i\\
        \implies & \frac{r - 1}{2H_n j - 1} \geq i. \tag{$2 H_n j - 1 > 0$ as $H_n \cdot j \geq 1$}\\
    \end{align*}
    Hence, $g(j) \le \frac{r - 1}{2H_n j - 1}$ for all $j$.
    In addition, since $H_n \cdot j \geq 1$, $H_n \cdot j \le 2H_n \cdot j - 1$,
    so $\frac{r - 1}{2H_n \cdot j - 1} \le \frac{r - 1}{H_n \cdot j} \le \frac{r}{H_n \cdot j}$.

    Plugging this into our earlier equation:
    \begin{align*}
        \sum_{i = 1}^n f(i) = \sum_{j = 1}^n g(j)
        \le \sum_{j = 1}^n \frac{r}{H_n \cdot j}= \frac{r}{H_n} \sum_{j = 1}^n \frac{1}{j}= r,
    \end{align*}
    as desired.
\end{proof}

We will now show the corresponding upper-bound for $k \ge n$, showing it is impossible to achieve $\alpha$-MMS for $\alpha > \frac{k}{H_n(m - n) - (m - k)}$. This is a generalization of a result by \citet{ABM16}. They show when rankings are complete, so $k = m$, no mechanism can guarantee $\alpha$-MMS for $\alpha > \frac{1}{H_n} + \e$ for any $\e > 0$. More specifically, in their model, such an approximation must hold for all $m$, so they in fact show it is impossible to achieve $\alpha > \frac{m}{(m - n)H_n}$ (which coincides with our bound when $m = k$), and take the limit as $m$ gets large. Our proof follows a similar structure to theirs with extra complexity due to the generalization.
	
	Suppose all agents agree on the top $k$ rankings, $g_1 \succ \ldots \succ g_k$ and do not rank the remaining $m - k$. First, we will show that if a mechanism achieves a positive MMS approximation on this allocation, the top $n$ goods $g_1, \ldots, g_n$ must be assigned such that each agent receives exactly one. Suppose this is not the case, which implies there is some agent $i$ who received none of these goods. In this case, it's possible that agent $i$ valued $g_1, \ldots, g_n$ at $\nicefrac{1}{n}$ and the remaining goods at zero, having a positive MMS, but receiving zero value. Therefore, without loss of generality, we can relabel the agents such that $g_i$ goes to agent $i$. 
	
	Fix the chosen allocation $A$.
	Let $r_i$ be the number of goods received by agent $i$ of the first $k$, $r_i = |A_i \cap \set{g_1, \ldots, g_k}|$, and let $q_i$ be the number of goods received by agent $i$ of the remaining $m - k$, $q_i = |A_i \setminus \set{g_1, \ldots, g_k}|$.
	We have that $\sum_{i = 1}^n r_i = k$ and $\sum_{i = 1}^n q_i = m - k$. It's consistent with the rankings that agent $i$ has the following valuations: $\nicefrac{1}{n}$ for goods $g_1, \ldots, g_{i - 1}$, $0$ for the $q_i$ non-top $k$ goods they received, $A_i \setminus \set{g_1, \ldots, g_k}$, and equal value $\frac{1}{m - q_i - i + 1}$ for the remaining goods. One possible partition for maximin share is to create $i - 1$ singleton bundles with goods $g_1, \ldots, g_{i - 1}$ and split the remaining $m - q_i - i + 1$ positively valued goods into $n - i + 1$ bundles as equally as possible. One of these bundles must have at most $\bigfloor{\frac{m - q_i - i + 1}{n - i + 1}}$ goods, however, agent $i$ received only $r_i$ of these items. Hence, if this allocation is $\alpha$-MMS, it must be the case that $p_i \ge \alpha \cdot \bigfloor{\frac{m - i - q_i + 1}{n - i + 1}}$. As $\sum_{i = 1}^n p_i = k$, we have that $\sum_{i = 1}^n \alpha \cdot \bigfloor{\frac{m - i - q_i + 1}{n - i + 1}} \le k$. From this, we get
	\begin{align*}
		\alpha
		&\le \frac{k}{\sum_{i=1}^n \bigfloor{\frac{m - i - q_i + 1}{n - i + 1}}}\\
		&\le \frac{k}{\sum_{i=1}^n \frac{m - i - q_i + 1}{n - i + 1} - 1}\\
		&= \frac{k}{\sum_{i=1}^n \frac{m - n - q_i}{n - i + 1}}\\
		&= \frac{k}{(m - n)\sum_{i=1}^n \frac{1}{n - i + 1} - \sum_{i=1}^n\frac{q_i}{n - i + 1}}\\
		&\le \frac{k}{(m - n) H_n - \sum_{i=1}^n\frac{q_i}{1}}\\
		&=  \frac{k}{(m - n) H_n - (m - k)}
	\end{align*}
	as we wanted to show.
    
	Next, let us consider the case of $k < n - 1$, where we will show that no allocations are $\alpha$-MMS for $\alpha > 0$.
	Consider a preference profile in which all agents agree on the top $k$ goods, that is, they rank $g_1, \ldots, g_k$, and do not rank the remaining $m - k$. Fix some arbitrary allocation $A$, we will show $A$ does not satisfy $\alpha$-MMS for any $\alpha > 0$.
	As $k < n - 1$, it must be the case that at least $n - k$ agents do not receive any of $g_1, \ldots, g_k$. Let $i$ be one of these agents that among them receives fewest goods, that is, $i \in \argmin_{j : A_j \cap \set{g_1, \ldots, g_k} = \emptyset} |A_j|$. We claim that $|A_i| \le m - n$. To see this, note that all of the at least $n - k$ agents without any of the first $k$ goods $g_1, \ldots, g_k$ have bundles of size at least $|A_i|$, so we have that $(n - k)|A_i| + k \le m$. Using the fact that $k = n - (n - k)$, after moving terms around, we get $|A_i| \le \frac{m - n}{n - k} + 1$. As $n - k > 1$ and $m - n > 0$, we get that $|A_i| < m - n + 1$, so since $|A_i|$ is an integer, it follows that $|A_i| \le m - n$. Therefore, it's consistent with the preference profile that agent $i$ has value $\nicefrac{1}{n}$ for $n$ of the goods not in $A_i$ including $g_1, \ldots, g_k$, and zero for all other goods. In this case, $\MMS_i$ is $\nicefrac{1}{n}$, but $i$ is receiving zero value from $A$, so $A$ cannot be $\alpha$-MMS for any positive $\alpha$.
	 
	Finally, let us now consider the case where $k = n - 1$. We will begin with upper-bounds, showing that no allocation can achieve higher than $1/ \bigfloor{\frac{m - n + 2}{2}}$-MMS\@. We will begin with a lemma.
	\begin{lemma}\label{lem:n-1-mms}
		When $k = n - 1$, if agents agree on the top $k$ goods, the only allocations that can satisfy $\alpha$-MMS for positive $\alpha$ are those where $n - 1$ of the agents get exactly one of $g_1, \ldots, g_k$, and the last agent receives all the remaining goods.
	\end{lemma}
	\begin{proof}
		Consider an arbitrary allocation $A$ where this is not the case. We will show $A$ cannot satisfy $\alpha$-MMS for positive $\alpha$. Since $A$ does not satisfy this, there must be some agent $i$ who receives none of $g_1, \ldots, g_k$ and also does not receive at least one good $g \in M \setminus \set{g_1, \ldots, g_k}$. It's consistent with the preference profile that $i$ has value $\nicefrac{1}{n}$ for $g_1, \ldots, g_k$ and $g$, and zero value for the remaining goods. Therefore ,$\MMS_i$ is $\nicefrac{1}{n}$, but under $A$, $i$ receive zero value, so $A$ cannot achieve a positive approximation.
	\end{proof}
	
	To show it is impossible to satisfy $\alpha$-MMS for $\alpha$ larger than $1/ \bigfloor{\frac{m - n + 2}{2}}$, suppose all agents agree on the rankings, so they report $g_1 \succ \ldots \succ g_k$ for the first $k$ goods and do not report the remaining $m - k$. Fix an arbitrary allocation $A$, we will show it cannot satisfy $\alpha$-MMS for $\alpha > 1/ \bigfloor{\frac{m - n + 2}{2}}$. By \Cref{lem:n-1-mms}, if $A$ does not give $n - 1$ agents exactly one of $g_1, \ldots, g_k$ and give the last agent all the remaining goods, it cannot have a positive MMS approximation, so suppose it does. Without loss of generality, suppose agent $i$ receives good $g_i$ for $i \in [n-1]$ and agent $n$ receives the remaining goods. In particular, consider agent $n - 1$, who received just good $g_k$. Suppose their valuations were as follows: they have value $\nicefrac{1}{n}$ for items $g_1, \ldots, g_{k-1}$ and equal value of $\frac{1}{m - n + 2}$ for all the remaining goods	. We have that the agent's MMS is at least $\bigfloor{\frac{m - n + 2}{2}} \cdot \frac{1}{m - n + 2}$ as they could create the partition with the first $n - 2$ sets being $\set{g_1}, \ldots, \set{g_{k - 1}}$ with the last two bundles created by splitting up the remaining goods as equally as possible, with sizes $\bigfloor{\frac{m - n + 2}{2}}$ and $\bigceil{\frac{m - n + 2}{2}}$. However, the agent only receives value, $\frac{1}{m - n + 2}$, so the MMS approximation is at most $1/ \bigfloor{\frac{m - n + 2}{2}}$-MMS\@.
	
	Next, we will show it is in fact possible to achieve $1/ \bigfloor{\frac{m - n + 2}{2}}$-MMS\@. The mechanism to achieve this will be the following picking sequence, the agents $1$ through $n$ followed by all $n$s, $1, 2, \ldots, n, n, \ldots, n$. We will show that every agent here receives a $1/ \bigfloor{\frac{m - n + 2}{2}}$ approximation to MMS\@.
	First, we will show the $n^{\text{th}}$ agent in fact receives an MMS allocation. Indeed, as the receives all but $n - 1$ of the goods, in any partition of the goods into $n$ bundles, at least one of the bundles must be a subset of these $m - (n - 1)$ goods. Hence, the value of them is an upper-bound on the agent's MMS\@. Next consider some agent $i \in [n - 1]$ that went $i^{\text{th}}$ in the order. Suppose their top $i$ favorite goods are $g_1, \ldots, g_i$. Note that they get at least their value for $g_i$. In any partition of the goods into $n$ bundles, at least $n - i + 1$ of them cannot contain their $g_1, \ldots, g_{i - 1}$. Further, as these $n - i + 1$ bundles collectively have at most $m - i + 1$ goods, at least one bundle must have at most $\bigfloor{\frac{m - i + 1}{n - i + 1}}$ goods, all of which are at most as valuable as $g_i$. Therefore, agent $i$ receives at least a $1/\bigfloor{\frac{m - i + 1}{n - i + 1}} \ge 1/\bigfloor{\frac{m - (n - 1) + 1}{n - (n - 1) + 1}} = 1/\bigfloor{\frac{m - n + 2}{2}}$ approximation to their MMS\@.
	\end{proof}

Strikingly, while $\Omega(n^2)$ distortion is unbeatable subject to $\alpha$-MMS for \emph{any} $\alpha > 0$, for $k \ge n$, we can achieve a matching $O(n^2)$ distortion even while simultaneously achieving the best-known MMS approximations for all $k$, introduced in \Cref{thm:mms-possible}. 

\begin{theorem}\label{thm:mms-distortion}
	The best possible distortion is as follows.
	\begin{itemize}
		\item On $\I^{k = n - 1}$, any deterministic ordinal rule satisfying $\alpha$-MMS for $\alpha > 0$ must have unbounded distortion. However, there is a randomized ordinal rule achieving the best-possible $1/ \bigfloor{\frac{m - n + 2}{2}}$-MMS and distortion $n$.
		\item On $\I^{k \ge n}$, any deterministic ordinal rule satisfying $\alpha$-MMS for $\alpha > 0$ must have distortion $\Omega(n^2)$, and there is a deterministic ordinal rule achieving $O(n^2)$ distortion with $\alpha$-MMS for the best-known $\alpha = \frac{k - n + 1}{m - n + 1} \cdot \frac{1}{2H_n}$. Further, there is a randomized ordinal rule that achieves $\alpha$-MMS for $\alpha = \frac{k - n + 1}{m - n + 1} \cdot \frac{1}{2H_n}$ with distortion $n$.
	\end{itemize}
\end{theorem}
\begin{proof}
	We will begin with the lower bound for distortion of deterministic rules on instances with $k = n - 1$. Fix some $n \ge 2$ and let $k = n - 1$. We will construct a family of instances with $n$ agents, top-$k$ rankings, increasing number of goods whose distortion grows linearly with the number of goods. This implies the distortion is unbounded. Suppose the agents agree on the ranking of the top $k$ goods $g_1 \succ \ldots \succ g_k$ and leave the remaining $m - k$ goods unranked. 
	Consider some allocation $A$ chosen by an $\alpha$-MMS rule for $\alpha > 0$. By \Cref{lem:n-1-mms}, under $A$, $n - 1$ agents must receive exactly one of $g_1, \ldots, g_k$. Without loss of generality, suppose agent $1$ receives $g_1$. Suppose the valuations are such that agent $1$ has value $\nicefrac{1}{m}$ for all the goods, while all other agents have value $1$ for $g_1$ and $0$ for all others. This means the social welfare is $\nicefrac{1}{m}$ while the highest possible social welfare is $1$, so the distortion is at least $m$. Further, this ratio holds even compared to the best $\alpha$-MMS allocation, as swapping $A_1$ and $A_2$ must remaining $\alpha$-MMS as all preference rankings are identical, while the social welfare is at least $1$.
	
	For randomized rules, however, note that $1/ \bigfloor{\frac{m - n + 2}{2}}$-MMS can be achieved by a picking sequence as shown in \Cref{thm:mms-possible}. By \Cref{prop:uniform-picking-sequences}, the uniform version of this picking sequence is also $1/ \bigfloor{\frac{m - n + 2}{2}}$-MMS and has distortion $n$.
	
	Next, we will consider the case where $k \ge n$. First, deterministic $\alpha$-MMS rules for $\alpha > 0$ must have distortion $\Omega(n^2)$ is an easy consequence of \Cref{thm:n2-lower}.
	
	We will now show it's possible to achieve $O(n^2)$ distortion even with an $\frac{k - n + 1}{m - n + 1} \cdot \frac{1}{2H_n}$-MMS deterministic rule. The construction is very similar to $k \ge n$ construction of \Cref{thm:mms-possible}, except that we add one more constraint: we want agent $1$ to also have additional appearances in the picking sequence, once at or before position $2 n j+1$ for each $j$. Further, if we are using top-$k$ rankings, we will give all the goods after the $k$ steps of the picking sequence to the agent $1$. Because agent $1$ gets the first overall pick as well, this ensures that the utility to agent $1$ is at least $\nicefrac{1}{2n}$ fraction of her value for all the goods, i.e., at least $\nicefrac{1}{2n}$. 
    Since $n$ is a trivial upper bound on the social welfare of any allocation, this implies the distortion upper bound of $O(n^2)$. 

    Now, as in the proof of \Cref{thm:mms-possible}, we must show that for all $d$, there are at most
    $d$ pairs with second coordinate at most $d$, also counting the additional pairs generated due to the extra constraint above. Since we have not added any pairs with second coordinate at most $n$, the part of the proof for $d \le n$ still holds. For $d \ge n + 1$, the number of new pairs added due to the above constraint is $\bigfloor{\frac{d - 1}{2n}}$. Hence, we need the following strengthening of \Cref{lem:d-ineq}.
    \begin{lemma}
        \label{lem:d-ineq-refinement}
        For all $n \in \bbN$ and for all $d \ge n + 1$,
        $\bigfloor{\frac{d - 1}{2n}} + \sum_{i = 1}^n \bigfloor{\frac{d - i}{2H_n (n - i + 1)}} \le d - n$.
    \end{lemma}
    \begin{proof}
		First note that when $n = 1$, this is equivalent to showing
		$\bigfloor{\frac{d - 1}{2}} + \bigfloor{\frac{d - 1}{2}} \le d - 1$ which trivially holds for all $d \in \bbN$. Therefore, from now on, we restrict to the case of $n \ge 2$.
		
		Similarly to \Cref{lem:d-ineq}, we consider two cases: $d \ge 4n$ and $d < 4n$. First, suppose $d \ge 4n$. Then,
		\begin{align*}
		\bigfloor{\frac{d - 1}{2n}} + \sum_{i=1}^n \bigfloor{\frac{d - i}{2H_n (n - i + 1)}}
		&\le \frac{d-1}{2n} + \sum_{i=1}^n \frac{d - i}{2H_n (n - i + 1)}\\
		&\le \frac{d}{4} + \sum_{i=1}^n \frac{d}{2H_n (n - i + 1)}\\
		&= \frac{d}{4} + \frac{d}{2H_n} \cdot \sum_{i=1}^n \frac{1}{n - i + 1}\\
		&= \frac{d}{4} + \frac{d}{2H_n} \cdot H_n\\
		&= \frac{3d}{4} \le d - n,
		\end{align*}
		where the second transition holds because $n \ge 2$, and the final transition holds because $d \ge 4n$.
		
		Next, suppose $d < 4n$. First, we tackle the special case of $n=2$. Note that when $d \le 2n$, we have $\bigfloor{\frac{d - 1}{2n}} = 0$, in which case the desired inequality is already implied by \Cref{lem:d-ineq}. For $2n < d < 4n$, i.e., $d \in \set{5,6,7}$, we have that $\bigfloor{\frac{d - 1}{2n}} = 1$. Hence, the desired inequality holds as follows:
		\begin{align*}
		d = 5: \quad & 1 + \sum_{i = 1}^2 \bigfloor{\frac{5 - i}{2H_2 (2 - i + 1)}} = 2 \le 5 - 1\\
		d = 6: \quad & 1 + \sum_{i = 1}^2 \bigfloor{\frac{6 - i}{2H_2 (2 - i + 1)}} = 2 \le 6 - 1\\
		d = 7: \quad & 1 + \sum_{i = 1}^2 \bigfloor{\frac{7 - i}{2H_2 (2 - i + 1)}} = 3 \le 7 - 1
		\end{align*}

		From now on, we restrict to the case of $n \ge 3$. Because we are in the case of $d < 4n$, this implies $\bigfloor{\frac{d - 1}{2n}} \le 1$. Hence, it is sufficient to show that 
		\[
		\sum_{i = 1}^n \bigfloor{\frac{d - i}{2H_n (n - i + 1)}} \le d - n - 1.
		\]
		As in \Cref{lem:d-ineq}, letting $r = d - n$, the above statement is equivalent to
		\[
		\sum_{i=1}^n \bigfloor{\frac{r + i - 1}{2H_n \cdot i}} \le r - 1.
		\]
		
		As in \Cref{lem:d-ineq}, define $f(i) = \bigfloor{\frac{r + i - 1}{2H_n \cdot i}}$ and $g(j) = \sum_{i = 1}^n \id[f(i) \ge j]$. In \Cref{lem:d-ineq}, we had $d < 2n$, giving us $f(i) \le r \le n$ for all $i \in [n]$. In this case, we still have that $f(i) \le r$ for all $i \in [n]$, but since we only have $d < 4n$, we have $r \le 3n$. We note that $g(j) \le \frac{r - 1}{2H_n j - 1}$ can still be derived as in the proof of \Cref{lem:d-ineq}. Thus, we have that
		\begin{align*}
		\sum_{i = 1}^n f(i) = \sum_{j = 1}^{3n} g(j) \le \sum_{j = 1}^{3n} \frac{r - 1}{2H_n j - 1} = (r - 1) \cdot \sum_{j = 1}^{3n} \frac{1}{2H_n j - 1}.
		\end{align*}
		
		To complete the proof, we simply need to show that $\sum_{j = 1}^{3n} \frac{1}{2H_n j - 1} \le 1$ for all $n \ge 3$. When $n = 3$, we have
		\begin{align*}
		\sum_{j = 1}^{3n} \frac{1}{2H_n j - 1} = \sum_{j = 1}^{9} \frac{1}{2H_3 j - 1} = \frac{19{,}657{,}653{,}727}{21{,}402{,}806{,}880} \le 1.
		\end{align*}
		
		Suppose $n \ge 4$. First, note that
		\begin{align*}
		\sum_{j = 1}^{3n} \frac{1}{2H_n j - 1} \le \sum_{j = 1}^{3n} \frac{1}{(2H_n - 1)j} = \frac{1}{2H_n - 1} \sum_{j = 1}^{3n} \frac{1}{j} = \frac{H_{3n}}{2H_n - 1},
		\end{align*}
		where the first transition holds because $j \ge 1$. To show that this final quantity is at most $1$, we need to show $H_{3n} \le 2H_n - 1$ for all $n \ge 4$. We prove this by induction on $n$. 
		
		For the base case of $n = 4$, we have that $H_{12} = \frac{86{,}021}{27{,}720} \le \frac{19}{6} = 2H_4 - 1$. Suppose $H_{3n} \le 2H_n - 1$ holds for some $n \ge 4$. We need to show that $H_{3(n+1)} \le 2H_{n+1} - 1$. We have that
		\begin{align*}
		H_{3(n+1)}
		&= H_{3n} + \frac{1}{3n + 1} + \frac{1}{3n + 2} + \frac{1}{3n + 3}\\
		&\le H_{3n} + \frac{3}{3n}\\
		&= 2H_n - 1 + \frac{1}{n}\\
		&\le 2H_n - 1 + \frac{2}{n + 1}\\
		&\le 2H_{n+1} - 1\qedhere
		\end{align*}
	\end{proof}
	Given \Cref{lem:d-ineq-refinement}, the result for deterministic rules follows.
	
	Finally, for randomized rules, note that the original rule of \cref{thm:mms-distortion} was a picking sequence, thus the uniform version has distortion $n$ by \Cref{prop:uniform-picking-sequences}.
\end{proof}

\section{Discussion}\label{sec:disc}

In this paper, we analyze which fairness properties can be achieved and what loss in social welfare must be incurred (distortion) when only ordinal preference information is provided in the form of top-$k$ rankings. 

This is inspired by a growing literature on distortion in voting~\cite{PR06}. A recent line of work has focused on imposing additional structure on the underlying cardinal preferences~\cite{ABEP+18,GHS20}. In fair division, one can also study natural restrictions on the underlying cardinal preferences such as a limit on the number of goods an agent can derive positive utility from or on the maximum difference between the values two agents can derive from the same good. 

Another thread of research on distortion in voting has focused on the tradeoff between distortion and the amount of preference information elicited~\cite{MPSW19,MPW20,Kem20,ABFV20}. Our results already offer one such tradeoff by allowing the designer to pick the value of $k$. An interesting direction for the future would be to study such a tradeoff while allowing arbitrary --- not necessarily ordinal --- elicitation and measuring the number of bits of information elicited. 

\section*{Acknowledgements}
Nisarg Shah was partially funded by an NSERC Discovery Grant.

\bibliographystyle{named}
\bibliography{abb,ultimate}

\appendix
\section*{Appendix}

\section{Impossible Fairness Properties}\label{sec:impossible-fairness}
Finally, we show that some fairness properties studied in the literature are impossible to guarantee given only the ordinal information, even with complete rankings. Consider an instance $I = (N, M, k, \sigma)$ with valuations $\vals$.

\begin{definition}\label{def:efx}
An allocation $\vec{A}$ is said to be \emph{envy-free up to any good (EFX)} if for every pair of agents
$i$, $j$ and every good $g \in A_j$, we have $v_i(A_i) \ge v_i(A_j \setminus \set{g})$.
\end{definition}

\begin{definition}\label{def:eq1}
An allocation $\vec{A}$ is said to be \emph{equitable up to one good (EQ1)} if for every pair of agents
$i$, $j$ such that $A_j \ne \emptyset$, there exists some good $g \in A_j$ such that $v_i(A_i) \ge v_j(A_j \setminus \set{g})$.
$\vec{A}$ is said to be \emph{equitable up to any good (EQX)} if for every pair of agents $i$, $j$ such that $A_j \ne \emptyset$
and for every good $g \in A_j$ such that $v_j(g) > 0$, we have that $v_i(A_i) \ge v_j(A_j \setminus \set{j})$.
\end{definition}

\begin{proposition}
    \label{prop:efx-impossible}
    There does not exist an ordinal allocation rule satisfying EFX even for two agents with access to complete rankings.
\end{proposition}
\begin{proof}
	Suppose there are two agents and four goods (denoted $g_1,\ldots,g_4$ for notational convenience). Consider the preference profile wherein both agents have the same preference ranking over the goods, given by $g_1 \succ g_2 \succ g_3 \succ g_4$.
	
	We show that regardless of the allocation chosen by a deterministic ordinal allocation rule, EFX cannot be satisfied with respect to every valuation profile consistent with this preference profile. 
	
	Suppose the rule picks allocation $\vec{A} = (A_1, A_2)$. There are two cases: $|A_1| \ne |A_2|$ and $|A_1| = |A_2|$.
	
	First, suppose $|A_1| \ne |A_2|$, and without loss of generality, suppose $|A_1| > |A_2|$. Since there are four goods, this implies that $|A_1| \ge |A_2| + 2$. In this case, if agent $2$ has equal value of $\nicefrac{1}{4}$ for all four goods, then agent $2$ would envy agent $1$ even after removing any single good from $A_1$, violating EFX (and in fact, EF1 too).
	
	Next, suppose that $|A_1| = |A_2|$, which implies that each is equal to $2$. Without loss of generality, suppose $g_1 \in A_1$. Then, consider the valuation function of agent $2$ which places value $1$ on good $g_1$ and $\epsilon < 1/3$ on every other good. Then, agent $2$ would envies agent $1$ even after removing the single good in $A_1 \setminus \set{g_1}$, violating EFX. 
\end{proof}

\begin{proposition}
    \label{prop:eq1-impossilbe}
    There does not exist an ordinal allocation rule satisfying EQ1 or EQX even for two agents with access to complete rankings.
\end{proposition}
\begin{proof}
	Consider the construction used in the proof of \Cref{prop:efx-impossible}, in which two agents have the same preference ranking over four goods given by $g_1 \succ g_2 \succ g_3 \succ g_4$.
	
	We show that regardless of the allocation chosen by a deterministic ordinal allocation rule, EQ1 cannot be satisfied with respect to all valuation profiles consistent with this preference profile. This implies the desired result for EQX as well. 
	
	Suppose the rule picks allocation $\vec{A} = (A_1, A_2)$.
	There are two cases: $|A_1| \ne |A_2|$ and $|A_1| = |A_2|$.
	
	First, suppose $|A_1| \ne |A_2|$, and without loss of generality suppose $|A_1| > |A_2|$. Since there are four goods, this implies that $|A_1| \ge |A_2| + 2$. Now, consider the valuation profile in which both agents have an equal value of $\nicefrac{1}{4}$ for all the goods. Then, even after removing any single good from $A_1$, the utility to agent $1$ will be more than the utility to agent $2$, violating EQ1. 
	
	Next, suppose $|A_1| = |A_2|$, which implies that each is equal to $2$.Without loss of generality, suppose $g_1 \in A_1$. Consider the valuation profile in which agent $2$ has value $1$ for $g_1$ and $0$ for every other good, while agent $1$ has value $\nicefrac{1}{4}$ for all the goods. Thus, the utility to agent $2$ is $0$, while the utility to agent $1$ remains higher than that, even after removing any single good from $A_1$, violating EQ1. 
\end{proof}

\end{document}